\documentclass[letterpaper, 10 pt, conference]{ieeeconf}  %

\IEEEoverridecommandlockouts%
\usepackage{graphicx}
\usepackage{float}
\usepackage{booktabs}
\usepackage{amsmath} %
\usepackage{amssymb}  %
\usepackage{mathtools}
\usepackage{xcolor}
\usepackage{cite}

\usepackage{amsthm}

\newtheorem{assumption}{Assumption}
\newtheorem{proposition}{Proposition}
\newtheorem*{remark}{Remark}
\newtheorem{theorem}{Theorem}
\newtheorem{lemma}{Lemma}
\newcommand{\R}{\ensuremath{\mathbb{R}}}
\DeclarePairedDelimiterX{\norm}[1]{\lVert}{\rVert}{#1}

\usepackage[ruled,vlined,linesnumbered]{algorithm2e}
\usepackage[ruled,linesnumbered]{algorithm2e} 
\SetKw{Continue}{continue}
\SetKw{Break}{break}
\SetKwFor{ForAll}{for all}{do}{end}
\SetKwFor{ParForAll}{for all}{in parallel with index $j$ do}{end}
\SetKwFor{ParForAllNoIndex}{for all}{in parallel do}{end}
\SetKwRepeat{Do}{do}{while}
\SetKwInOut{Parameter}{Parameters}
\SetKwComment{Comment}{$\triangleright$\ }{}

\title{\LARGE \bf
    Bayesian Safe Learning and Control with Sum-of-Squares Analysis and Polynomial Kernels
}

\author{Alex Devonport, He Yin, and Murat Arcak%
\thanks{Alex Devonport and Murat Arcak are with the Department of Electrical
Engineering and Computer Sciences, University of California, Berkeley
\texttt{\{alex\_devonport, arcak\}@berkeley.edu}}%
\thanks{He Yin is with the Department of Mechanical Engineering, University of
California, Berkeley \texttt{he\_yin@berkeley.edu}}%
}

\begin{document}

\maketitle
\thispagestyle{empty}
\pagestyle{empty}

\begin{abstract}

We propose an iterative method to safely learn the unmodeled dynamics of a nonlinear system using Bayesian Gaussian process (GP) models with polynomial kernel functions. 
The method maintains safety by ensuring that the system state stays within the region of attraction (ROA) of a stabilizing control policy while collecting data. 
A quadratic programming based exploration control policy is computed to keep the exploration trajectory inside an inner-approximation of the ROA and to maximize the information gained from the trajectory.
A prior GP model, which incorporates prior information about the unknown dynamics, is used to construct an initial stabilizing policy.
As the GP model is updated with data, it is used to synthesize a new policy and a larger ROA, which increases the range of safe exploration. 
The use of polynomial kernels allows us to compute ROA inner-approximations and stabilizing control laws for the model using sum-of-squares programming.
We also provide a probabilistic guarantee of safety which ensures that the policy computed using the learned model stabilizes the true dynamics with high confidence.

\end{abstract}

\section{Introduction}

Learning-based methods allow for the control of systems for which accurate or analytically
tractable models are not available. The learning method constructs a model using data collected from the
system; however, for safety-critical systems this data must be collected in such
a way that the system is not put in danger.

The practice of collecting data from a dynamical system for a learning model
while keeping the system safe is called \emph{safe learning}, and has been
investigated both from the perspective of both robust control 
\cite{jin2018stability, vinogradska2017stability} 
and reinforcement learning 
\cite{chow2018lyapunov, chow2019lyapunov, akametalu2014reachability, 2019Chengend2end}.
Several control-theoretic guarantees may be used to certify safety.
These include Lyapunov functions for the learned model
\cite{berkenkamp2016safe, berkenkamp2016safe2, berkenkamp2017safe, richards2018lyapunov, khansari2014learning, 2019Gallieri, Mohammad2014752}
to ensure stability,
barrier functions to guarantee that the states remain in an invariant set
\cite{wang2018safe, taylor2019learning, ahmadi2020safe, 2019Fan},
and reachability methods such as Hamilton-Jacobi analysis
to ensure that the states can reach a target set and
avoid unsafe sets despite model inaccuracies
\cite{fisac2018general, akametalu2014reachability}.

Gaussian processes (GPs) are a popular model for incorporating learning-based methods
into the analysis of control systems, in particular for safe learning. 
Unlike many learning models, GP models have a closed-form expression for predictions, as well as a quantification of prediction uncertainty. This allows for control-theoretic guarantees to be applied to models learned by a GP.
For example, 
\cite{berkenkamp2016safe, berkenkamp2016safe2}
use a Lyapunov approach to guarantee that a partially unknown system with a
fixed policy is stable with high probability, and to compute a region of attraction.
This approach works by verifying that the Lyapunov condition holds with high
probability for the GP model over a grid of points in the state space. A theorem
in 
\cite{srinivas2009gaussian}
implies under certain conditions that this guarantee on the GP model holds
for the true dynamics as well.
While this approach is effective for verifying the stability of the learned
model, it requires a base stabilizing policy and Lyapunov function. The policy and Lyapunov function stay fixed,
and cannot be improved as a more accurate model is learned.

To lift the fixed Lyapunov function restriction, \cite{umlauft2017learning}
learns a GP state space model for a globally stable system, as well as a sum-of-squares (SOS) Lyapunov function. Since the learned GP model is not necessarily stable, the learned Lyapunov function is used to further stabilize the GP model. However, safe exploration is not considered, and the Lyapunov condition is only loosely enforced on a finite set of states, rather than being guaranteed.
To lift the fixed policy restriction, several recent works use reinforcement learning. In this approach, an iterative reinforcement
learning algorithm like policy gradient
\cite{berkenkamp2017safe} 
or imitation learning 
\cite{taylor2019learning, khansari2014learning} 
is verified at each
iteration by Lyapunov analysis. While this approach allows for policy
optimization, it still relies on a given, fixed Lyapunov function to verify
safety.

In this paper, we propose an algorithm that avoids both restrictions by using a GP model with polynomial kernel functions to model the unknown dynamics.
Polynomial kernel functions allows us to use SOS techniques \cite{Jarvis:05} to synthesize a stabilizing control policy and an inner-approximation to the region of attraction (ROA).
The algorithm uses the GP model in a Bayesian framework for learning: this allows us to incorporate information about the dynamics into a prior model, which we update using data to form a posterior model.
Since the prior and posterior models are both GPs, we can use the same SOS techniques to provide a safety guarantee for both the prior model and the posterior model. We also propose an \emph{exploration policy} which allows for safe exploration inside the ROA inner-approximation to increase the information gained from system trajectories.

We also provide the following theoretical results. First, we establish a probabilistic safety guarantee, which under appropriate system conditions ensures that the stabilizing policy computed for the learned GP model also stabilizes the true system, and that the ROA computed for the learned system is an inner-approximation of the true ROA, with high probability. This result is described in Theorem~\ref{thm:ROA_them}. Second, we construct a polynomial kernel function suitable for modeling uncertain dynamics around a known equilibrium.

\subsection*{Notation}

The subscript $x_i$ denotes the $i^{th}$ element of the vector $x$. The superscript $x^{(i)}$ with parentheses denotes the data point in the data set $\mathcal{D}$ with index $i$. The superscript $x^{i}$ without parentheses denotes an object associated with the $i^{th}$ iteration of an algorithm.

When applied to vectors, the orders $>$, $\le$ are applied elementwise. The operator $\mathbb{E}[\cdot]$ denotes expectation with respect to a probability distribution.

For $\xi \in \mathbb{R}^n$, $\mathbb{R}[\xi]$ represents the set of polynomials in $\xi$ with real coefficients, and $\R^m[\xi]$ and $\R^{m \times p}[\xi]$ denote all vector- and matrix-valued polynomial functions. The subset $\Sigma[\xi] := \{\pi = \sum_{i=1}^M \pi_i^2: \pi_1,...,\pi_M \in \mathbb{R}[\xi]\}$ of $\mathbb{R}[\xi]$ is the set of SOS polynomials in $\xi$. 

\section{Preliminaries}%
\label{sec:preliminaries}

Consider a continuous-time nonlinear system of the form
\begin{equation}
    \label{eq:true_dynamics}
    \dot{x}(t) = f(x(t)) + g(x(t))u(t) + w(x(t)), 
\end{equation}
with state $x(t) \in\R^{n_x}$ and input $u(t) \in\R^{n_u}$.
The system dynamics comprise a known control-affine part, $f$ and $g$, and an unknown term $w$ which must be learned.

We will assume the true model has a known equilibrium at the origin, so that a stabilizing policy and region of attraction (ROA) can be constructed.
\begin{assumption}%
    \label{a:origin_is_equilibrium}
    The origin $(x=0,u=0)$ is an equilibrium of~\eqref{eq:true_dynamics}, that is
    $f(0)=w(0)=0$.
\end{assumption}
To allow for sum-of-squares (SOS) analysis, we make the following assumption:
\begin{assumption}%
    \label{a:polynomial_baseline_model}
    The known dynamics are polynomials: $f(x) \in \R^{n_x}[x]$ and $g(x) \in \R^{n_x \times n_u}[x]$. 
\end{assumption}
The true dynamics need not be polynomial, as the non-polynomial terms can be
absorbed into $w(x)$. 
We also assume that the unknown term can be approximated by a polynomial-kernel Gaussian process (GP).
\begin{assumption}%
    \label{a:uncertainty_approximated_by_polynomial}
    The term $w(x)$ can be approximated by a polynomial in a region $\mathcal{X}\in\R^{n_x}$ containing
    the origin. Specifically, for a given $\epsilon>0$ there is a polynomial
    $q(x)$ such that $\norm{w(x)-q(x)} \le \epsilon$ for all $x\in\mathcal{X}$.
\end{assumption}
For example, if $w$  is analytic in a ball $\mathcal{B}$ containing the origin, then Taylor's theorem ensures that Assumption \ref{a:uncertainty_approximated_by_polynomial} holds in $\mathcal{B}$.

Aside from any prior knowledge, our information about the system will come from measurements of the form $(x^{(i)},u^{(i)}, \dot{x}^{(i)})$. Typically $\dot{x}$ itself is not directly measurable, and is estimated using a finite-difference approximation from measurements of $x$. The finite-difference approximation will be a noisy estimate of $\dot{x}$, and the measurements of $x$ may in practice be noisy as well. 
\begin{assumption}
\label{a:bounded_noise}
We have access to measurements of $\dot{x}$ which are corrupted by noise which is uniformly bounded by $\sigma_n$.
\end{assumption}

Our analysis has three goals. The first is 
to use data collected from system trajectories
to model the unknown part of the dynamics.
The second is to use the learned model to synthesize a stabilizing controller for the system and a ROA inner-approximation which holds for the true dynamics with high probability. The third is to design an exploration controller to maximize the information collected during the exploration trajectory while maintaining it inside the ROA.

\begin{remark}
In \eqref{eq:true_dynamics}, we assume $w$ depends only on $x$. If it depends both on $x$ and $u$, we introduce an auxiliary input state $x_u(t) \in \R^{n_u}$ for $u$, and design the new input $v(t) \in \R^{n_u}$ for $x_u$. This leads to the augmented system
\begin{align}
    \dot{x}(t) &= f(x(t)) + g(x(t))x_u(t) + w(x(t),x_u(t)) \nonumber \\
    \dot{x}_u(t) &= v(t), \nonumber
\end{align}
which recovers the form in \eqref{eq:true_dynamics}. This formulation is demonstrated in Section~\ref{subsec:example_inverted_pendulum}.
\end{remark}

\section{Estimating the Unmodeled Dynamics}%
\label{sec:estimating_the_unmodeled_dynamics}

To estimate the unknown term in a Bayesian framework, we must choose a
\emph{prior distribution} for the system dynamics. The prior model is a probability distribution of candidate functions for $w$, which represents what we know
about the system prior to seeing any data. From Assumption~%
\ref{a:uncertainty_approximated_by_polynomial} we
know that the system can be approximated by a polynomial in a region about the
equilibrium, so we will choose a prior over polynomial functions. Assumption~\ref{a:polynomial_baseline_model} implies that $f(x)+g(x)u$ is an estimate for the true dynamics: assuming this is the best estimate we can make without data, we will
take the prior mean for $w$ to be zero. 

We will use a GP as our prior distribution. A GP $h$
is a probability distribution over functions which is completely characterized by
its mean $m(x) = \mathbb{E}[h(x)]$ and covariance 
$k(x,y)=\mathbb{E}[(h(x)-m(x))(h(y)-m(y))]$. The covariance of a GP prior is also called the \emph{kernel function} of the process. The kernel function determines the class of
functions over which the distribution is defined.
When $k(x,y)$ is polynomial in
$x$ and $y$, the distribution will be over a space of polynomial functions.
We will therefore choose $k(x,y)$ to be a polynomial.

Typically, GPs are presented as distributions of scalar-valued functions. Since the unknown term $w$ is vector-valued, we will model each entry $w_i$ with a separate scalar-valued GP of functions with domain $\R^{n_x}$.
We write our prior distribution for the dynamics as
\begin{equation}
    \label{eq:prior_dynamics}
    \dot{x}(t) = f(x(t)) + g(x(t))u(t) + \hat{w}(x(t)),
\end{equation}
where $\hat{w}(x)$ is a vector of GPs $\hat{w}_i$, each with mean zero and kernel $k_i$.

As we collect data from a
system trajectory, we condition the prior distribution on the data to obtain the 
\emph{posterior distribution}. Like the prior, the posterior is a distribution over functions. 
For a GP prior, the posterior will also be a GP, but with a different mean and covariance which more accurately
represent the ground truth than the prior. 

\subsection{GPs with polynomial kernels}

Consider a scalar GP prior $h$ with mean zero and kernel $k(x,y)$, and a data set $\mathcal{D} = \{(x^{(i)}, y^{(i)})\}_{i=1}^N$ of states $x^{(i)}\in\R^{n_x}$ and labels $y^{(i)}\in\R$. Then the posterior distribution, that is the prior conditioned on the data, is also a GP, whose mean and variance have closed-form solutions~\cite{gpml}. The posterior mean has the form
\begin{equation}
\label{eq:gp_mean}
m(x) = \mathbb{E}[h(x) | \mathcal{D}] = y^\top(K+\sigma_nI)^{-1}k_* 
\end{equation}
where $K$ is the kernel Gramian matrix with elements $(K)_{ij} = k(x^{(i)},x^{(j)})$, $k_*$ is the vector with elements $(k_*)_i = k(x,x^{(i)})$, and $y$ is the vector with $y_i=y^{(i)}$. Letting $c=y^\top(K+\sigma_nI)^{-1}$, we can re-express the mean as
\begin{equation}
    \label{eq:posterior_mean_sum}
    m(x) = \sum_{i=1}^N c_i k(x,x^{(i)}).
\end{equation}
When $k(x,y)$ is a polynomial in $x$ and $y$, (\ref{eq:posterior_mean_sum}) shows that $m(x)$ is also a polynomial, of the same degree as the kernel.

The posterior variance has the form
\begin{equation}
\begin{split}
\label{eq:gp_var}
    Var(x)&=
    \mathbb{E}[(h(x)-m(x))^2|\mathcal{D}]\\
    &= k(x,x) - k_*^\top(K+\sigma_nI)^{-1}k_*,
\end{split}
\end{equation}
which is a polynomial when $k(x,y)$ is a polynomial.
The degree of $Var(x)$ will be twice the degree of $m(x)$.%

\subsection{Choice of polynomial kernel}

The spaces of polynomials from which the mean and variance are drawn depend on the specific choice of polynomial kernel. In particular, the mean is drawn from the \emph{reproducing kernel Hilbert space} (RKHS) $\mathcal{H}(k)$ of the kernel $k$. The kernel must be chosen so that the functions in $\mathcal{H}(k)$ satisfy 
Assumption \ref{a:origin_is_equilibrium}.
To construct a suitable kernel, we use two classic results which follow from~\cite{aronszajn1950theory}:
\begin{proposition}
\label{prop:rkhs_homogeneous}
The RKHS $\mathcal{H}(k)$ of the \emph{homogeneous polynomial kernel} $k(x,y)=\alpha^2(x^\top y)^p$ is spanned by the monomials of degree $p$, that is by monomials $\prod_{i=1}^d x_i^{p_i}$ such that $\sum_i p_i=p$.
\end{proposition}
Here, $\alpha^2$ and $p$ are hyperparameters: $\alpha^2$ is a scaling factor, and $p$ sets the polynomial degree.
\begin{proposition}
\label{prop:rkhs_sum}
Let $k_1$ and $k_2$ be two kernels of finite-dimensional RKHSs. Then $k_1+k_2$ is also a kernel, and $\mathcal{H}(k_1+k_2)$ is spanned by the concatenation of the spans of $\mathcal{H}(k_1)$ and $\mathcal{H}(k_2)$.
\end{proposition}
For example, the function $(x^\top y)^2+(x^\top y)^3$ is a kernel function whose RKHS is spanned by the monomials of degrees 2 and 3. This motivates the following choice of kernel:
\begin{equation}
\label{eq:poly_kernel_zao}
    k(x,y) = \alpha_1^2(x^\top y) + \alpha_2^2(x^\top y)^2 + \dotso + \alpha_p^2(x^\top y)^p.
\end{equation}
By Propositions \ref{prop:rkhs_homogeneous} and \ref{prop:rkhs_sum}, the RKHS of this kernel is spanned by all monomials of degree $\le p$ \emph{except} for degree zero. In other words, the RKHS spans all polynomials $q$ of degree $\le p$ that satisfy $q(0)=0$.

In Section \ref{subsec:probabilistic_bounds} we will see that the range of possible unknown terms admitted by the prior model is bounded with high probability by a multiple of $\sqrt{Var(x)}$, so that a higher variance admits a larger class of functions for the unknown term. 
Therefore, any prior knowledge about the general form of the unknown term or the range of values it can take on should be used to select the kernel. While keeping the form $\eqref{eq:poly_kernel_zao}$, this information can be used to choose the hyperparameters $\alpha_i^2$.
For instance, if the baseline $f$ and $g$ are known to be accurate up to degree 2, then $\alpha_1^2$ and $\alpha_2^2$ can be set to small values, while the other $\alpha_i$ are set to high values. Another example is if the dynamics are known \emph{a priori} to be even (or odd); then, the prior kernel need only contain terms of even (or odd) degree. %

\subsection{Probabilistic Bounds on the GP Model}
\label{subsec:probabilistic_bounds}
The following inequality from \cite{srinivas2009gaussian} provides a probabilistic bound on the values that the functions in the distribution of a GP can take over its domain.
\begin{lemma}[Theorem 6 of \cite{srinivas2009gaussian}]
\label{lem:krause_inequality}
Suppose we have data $\{x^{(i)},y^{(i)}\}_{i=1}^N$ from a function $h\in\mathcal{H}(k)$ that satisfies $\norm{h}_k\le\infty$, where $\norm{\cdot}_k$ is the norm of $\mathcal{H}(k)$. The data may be corrupted with noise uniformly bounded by $\sigma_n$. Let $\beta_N = 2\norm{h}_k^2 + 300\gamma_N\log^3(N/\delta)$, where $\gamma_N$ is the maximum mutual information that can be obtained for the GP prior with $N$ samples corrupted with noise bounded by $\sigma_n$. Let $\delta\in(0,1)$. Then the inequality
\begin{equation} \label{eq:hhatbound}
    |h(x)-m_{\hat{h}}(x)| \le \sqrt{\beta_N} \sigma_{\hat{h}}(x)
\end{equation}
holds with probability $\ge 1 - \delta$, where $m_{\hat{h}}(x)$ and $\sigma_{\hat{h}}(x)=\sqrt{Var_{\hat{h}}(x)}$ are the mean and standard deviation of the GP $\hat{h}$ with kernel function $k$ conditioned on the data.
\end{lemma}
The inequality \eqref{eq:hhatbound} transforms the problem of providing a probabilistic guarantee for a GP into the problem of providing a guarantee over functions with a given upper and lower bound. In Section \ref{sec:estimating_the_region_of_attraction}, we will show that for GPs with polynomial kernels, this further transforms into a problem that may be solved with SOS programming.

The assumptions we have made on the system allow us to use this inequality in our analysis.
\begin{proposition}
\label{prop:krause_inequality_holds}
For a kernel $k(x,y)$ of the form (\ref{eq:poly_kernel_zao}) with sufficiently high degree $p$, measurements of $w$ can be used to construct a GP model which satisfies the inequality in Lemma \ref{lem:krause_inequality}.
\end{proposition}
Proposition \ref{prop:krause_inequality_holds} is proved in Appendix~\ref{app:prfprop4}. 

The quantity $\gamma_N$ is difficult to compute exactly for most kernels, and differs for each data set size $N$. However, for many commonly-used kernels it has a sublinear dependence on $N$, and can be effectively approximated up to a constant \cite{srinivas2009gaussian}. We will assume through the rest of the paper that the quantity $\sqrt{\beta_N}$ can be bounded by a constant parameter $\eta$. The parameter $\eta$ is higher for smaller values of $\delta$, i.e. for probabilistic bounds of higher confidence.

\section{Estimating the Region of Attraction}%
\label{sec:estimating_the_region_of_attraction}

For safe learning with a GP model we must ensure that there is a region of state space which we are confident can be explored safely. To do this, we will synthesize a memoryless, state feedback control policy $\kappa$ and a
Lyapunov function $V$ which guarantee that the closed-loop system
\begin{equation}
    \label{eq:true_cl_dynamics}
    \dot{x}(t) = f(x(t)) + g(x(t))\kappa(x(t)) + w(x(t))
\end{equation}
is stable around the origin with high confidence according to our model. By this, we
mean that the closed-loop system satisfies the inequality
\begin{equation}
    \label{eq:probabilistic_vdot_inequality}
    \tfrac{\partial V(x)}{\partial x} \cdot (f(x) + g(x)\kappa(x) + w(x)) < 0
\end{equation}
with probability $\ge 1-\delta$, with $\delta\in(0,1)$, for all points in a set $\mathcal{R} \backslash 0$, where
\begin{align}
\label{eq:probabilistic_roa}
\mathcal{R}=\{x \in \R^{n_x} \vert V(x) \le\gamma\}, \ \text{for some} \ \gamma > 0.
\end{align}
This set is an inner-approximation of the
origin's ROA, which we will make as large as possible.

We can ensure that the inequality (\ref{eq:probabilistic_vdot_inequality}) holds with high probability on the true dynamics by ensuring it holds for the deterministic bounds $|w_i(x)-m_i(x)| \le \eta \sigma_i(x)$, $i=1,\dotsc,n_x$, which are derived from Lemma \ref{lem:krause_inequality}. 
The link between the deterministic bound on the GP model and the probabilistic guarantee of stability for the true dynamics is stated in the following theorem.

\begin{theorem} \label{thm:ROA_them}
Suppose the true dynamics satisfy the assumptions outlined in Section \ref{sec:preliminaries}. Let $\eta$ be a bound on the parameter $\sqrt{\beta_N}$ such that Lemma \ref{lem:krause_inequality} holds for each $w_i$ with a given $\delta\in(0,1)$. Let $\sigma(x)$ be the vector of standard deviations $\sigma_i(x)=\sqrt{Var_i(x)}$, where $Var_i(x)$ is the variance of the GP $\hat{w}_i$.
Given $f,g$ defined in \eqref{eq:true_dynamics}, and $\gamma > 0$, if there exists a control law $\kappa:\R^{n_x} \rightarrow \R^{n_u}$, and a $\mathcal{C}^1$ function $V: \R^{n_x} \rightarrow \R$, such that $V(0)=0$ and $V(x)>0$ for all $x\in\R^{n_x}\backslash 0$, and
\begin{equation}
    \label{eq:Vdot}
    \tfrac{\partial V(x)}{\partial x} \cdot (f(x)+g(x)\kappa(x)+m(x)+d(x)) <0
\end{equation}
holds in a bounded region $\mathcal{R}\backslash 0 \subset \R^{n_x}$ for all vector functions $d$ bounded by $-\eta\sigma(x)\le d(x) \le \eta\sigma(x)$,
then $\mathcal{R}$ is an inner-approximation to the ROA of \eqref{eq:true_dynamics} with probability $\ge 1-\delta$ over the GP distribution.
\end{theorem}
\begin{proof}
Proposition \ref{prop:krause_inequality_holds} establishes that Lemma \ref{lem:krause_inequality} holds for the true dynamics and the GP model. Therefore, the bounds $|w_i(x)-m_i(x)| \le \eta\sigma_i(x)$, and equivalently the bound $-\eta\sigma(x)\le d(x) \le \eta\sigma(x)$, hold with probability $\ge 1-\delta$. Since $V$ guarantees that (\ref{eq:Vdot}) holds for $-\eta\sigma(x)\le d(x) \le \eta\sigma(x)$, it follows by Lemma \ref{lem:krause_inequality} that the same $V$ ensures that (\ref{eq:probabilistic_vdot_inequality}) holds with probability $\ge 1-\delta$ for $x\in\mathcal{R}\backslash 0$. This ensures that $\mathcal{R}$ is a ROA inner-approximation for the true dynamics with probability $\ge 1-\delta$.
\end{proof}

We restrict decision variables $V$ and $\kappa$ to be polynomials in order to use SOS analysis~\cite{Jarvis:05} to synthesize them such that the condition \eqref{eq:Vdot} holds. The condition \eqref{eq:Vdot} is a set containment constraint, and the generalized S-procedure \cite{Parrilo:00} can be used to derive the corresponding SOS constraint for it. To do this, we must express the bound $-\eta\sigma(x) \le d(x) \le \eta\sigma(x)$ as a semi-algebraic set. The bound can be %
described by a number of $n_x$ quadratic constraints: for $i = 1,...,n_x$,
\begin{align}
\eta^2 \sigma_i(x)^2 - d_i^2(x) = \eta^2 Var_i(x) - d_i^2(x) \ge 0 \ \forall x, \label{eq:choice2}
\end{align}
which use the polynomial $Var$ directly. Define polynomials $p_{d,i}(x,d) = \eta^2 Var_i(x) - d_i^2$ for $i = 1,...,n_x$. By choosing the volume of $\mathcal{R}$ as the reward function to be maximized, and applying the generalized S-procedure to  \eqref{eq:Vdot}, we obtain the following SOS optimization problem (dropping dependence on $x$ and $t$ for compactness of notation):
\begingroup
\allowdisplaybreaks
\begin{subequations}\label{eq:sos_ROA}
\begin{align} 
	\sup_{V,\kappa,s} &\text{Volume}(\mathcal{R}) \nonumber \\
	\ \text{s.t.} \ \ & s_V, s_{d,i} \in \Sigma[(x,d)], \nonumber \\
	& V - \epsilon_1 x^\top x \in \Sigma[x], \kappa \in \R^{n_u}[x],  \\
	& -(\tfrac{\partial V}{\partial x} \cdot (f+g\kappa + m + d) + \epsilon_2 [x;d]^\top[x;d])\nonumber \\
	& \quad \quad  + (V - \gamma)s_V - \sum_{i=1}^{n_x} s_{d,i} p_{d,i} \in \Sigma[(x,d)],
\end{align}
\end{subequations}
\endgroup
where $\epsilon_1$ and $\epsilon_2$ are small positive numbers. The optimization~\eqref{eq:sos_ROA} is a non-convex problem, since it is bilinear in two sets of decision variables $V$ and $(s_V, \kappa)$. It can be handled by alternating the search over these two sets of decision variables, since holding one set fixed while optimizing over the other results in a convex problem. The procedure is summarized in Algorithm~\ref{alg:solve_sos} in Appendix~\ref{app:iteralgo}.

\section{Exploring the Region of Attraction}%
\label{sec:exploring_the_region_of_attraction}

In order to increase the information gained from the trajectory data, we would like for each trajectory to explore a different region of the state space while remaining in the ROA inner-approximation. While the policy $\kappa$ synthesized from SOS programming ensures that system stays in the inner-approximation, it does not ensure that new areas of the state space will be explored. Therefore, as an alternative to the the control policy $\kappa(x)$, we propose an \emph{exploration policy} $\kappa_e(x)$ which guides the system to areas of the state space with little data.

The posterior variances $Var_i(x)$ of the GPs $\hat{w}_i$ can be used to track which areas of the state space have not been visited. In regions close to a data point, $Var_i(x)$ will be close to the noise level $\sigma_n$; in regions far from any data, $Var_i(x)$ will be close to the prior variance $k(x,x)$. Therefore, guiding the system to areas of high variance will lead it to areas which have not been explored. We can ensure this by choosing $\kappa_e$ to increase $Var_i(x(t))$, the variance of $\hat{w}_i$ at the current state, over time. To account for each $\hat{w}_i$, we will try to increase the sum $\sum_i Var_i(x(t))$.

The exploration policy will choose a control action by solving an optimization problem. The problem will be to maximize the time derivative of $\sum_i Var_i(x(t))$, the sum of the variances at the present system state. The derivative is
(dropping dependence on $x$ and $t$ for compactness of notation)
\begin{equation}
\begin{split}
    \tfrac{d}{dt} \sum_{i=1}^{n_x} Var_i &= \sum_{i=1}^{n_x} \tfrac{\partial Var_i}{\partial x} \dot{x} \\
    &= \left(\sum_{i=1}^{n_x} \tfrac{\partial Var_i}{\partial x}\right)(f+g\kappa_e + \hat{w}).
\end{split}
\end{equation}
To maximize this expression using $\kappa_e(x)$ as a decision variable, we need only consider the $(\sum_i \tfrac{\partial Var_i}{\partial x})g\kappa_e$ term.

At the same time, $\kappa_e(x)$ must not take the system outside of the ROA. Therefore, $\kappa_e$ must satisfy
\begin{align}
    \dot{V}_{low} &= \tfrac{\partial V}{\partial x}\cdot 
    (f + g\kappa_e + m - \eta \sigma) \le 0 \\
    \dot{V}_{up} &= \tfrac{\partial V}{\partial x}\cdot 
    (f + g\kappa_e + m + \eta \sigma) \le 0.
\end{align}
To ensure a unique solution, we will also include a quadratic regularizing term on $\kappa_e$ in the objective.
The form of the exploration policy $\kappa_e$ is then
\begin{equation}
\label{eq:exploration_policy}
\begin{aligned}
&\kappa_e(x) = &\text{arg }\underset{u}{\text{max}}
 & \left(\sum_{i=1}^{n_x} \tfrac{\partial Var_i}{\partial x}\right)g u
 - \lambda u^\top u \\
& & \text{s.t. }
 & \tfrac{\partial V}{\partial x}\cdot(f + g u + m - \eta \sigma) \le 0\\
 & & & \tfrac{\partial V}{\partial x}\cdot(f + g u + m + \eta \sigma) \le 0,
\end{aligned}
\end{equation}
where $\lambda > 0$ is a regularization parameter. The policy (\ref{eq:exploration_policy}) is a quadratic program, since for a fixed $x$ the objective is quadratic and the constraints linear in $u$.
Since quadratic programs can be efficiently solved in real time,
the exploration policy is suitable for online use.

When the policy $\kappa$ from (\ref{eq:sos_ROA}) exists, $u=\kappa(x)$ is a feasible solution to (\ref{eq:exploration_policy}). This means that (\ref{eq:exploration_policy}) is feasible when \eqref{eq:sos_ROA} is feasible.

\section{An Algorithm for Safe Learning}%
\label{sec:an_algorithm_for_safe_learning}

Algorithm \ref{alg:sos_gp} below shows how the results of sections \ref{sec:estimating_the_unmodeled_dynamics}, \ref{sec:estimating_the_region_of_attraction}, and \ref{sec:exploring_the_region_of_attraction} can be combined to perform safe exploration and robust policy synthesis.

The first step is to establish the prior information available for the system dynamics and encode it into a prior model. This comprises choosing the terms $f$ and $g$ in the control-affine base model, and selecting a prior kernel $k(x,y)$ for the unknown term. The base model $f$ and $g$ may come, for example, from a linearized model of the system. The kernel should be chosen so as to capture any further knowledge about the unknown part of the dynamics. 

With the prior model in place, the next step is to synthesize a prior control policy $\kappa^0$, a prior Lyapunov function $V^0$ and a prior $\gamma^0$ by solving the SOS program (\ref{eq:sos_ROA}) using Algorithm~\ref{alg:solve_sos}.
The prior Lyapunov function acts as a certificate that $\kappa^0$ stabilizes the equilibrium with high probability, that is for a large probability mass of candidates for $w$ admitted by the prior model. Sublevel sets of $V^0$ also act as inner-approximations of the ROA created by $\kappa^0$. We take the prior ROA as the sublevel set $\mathcal{R}^0=\{x\in\R^{n_x}|V^0(x) \le\gamma^0\}$. %

After synthesizing the prior ROA inner-approximation, the next step is to collect data to form the posterior model. This data will come from a system trajectory whose initial condition we may choose. In order to collect data safely, we choose an initial condition inside the prior ROA estimate (step \ref{ln:select_initial_condition}), so that the system is guaranteed to eventually return to the origin. Rather than use the prior policy $\kappa^0$, we will use the exploration policy $\kappa_e$ to guide the trajectory of the system during data collection (step \ref{ln:get_data}). This will ensure that the system trajectory visits regions where model variance (i.e. uncertainty) is high.

After collecting data, the next step (step \ref{ln:do_gp_regression}) is to compute the posterior model using (\ref{eq:gp_mean}) and (\ref{eq:gp_var}). With the posterior model, we can solve the SOS program (\ref{eq:sos_ROA}) using the posterior model (step \ref{ln:do_posterior_sos_analysis}) to synthesize a posterior policy $\kappa^1$ and posterior Lyapunov function $V^1$, and compute an inner-approximation of the ROA for the posterior policy. Since $\kappa^1$ and $V^1$ are computed using a more accurate model of the dynamics, the posterior ROA estimate will generally be larger than the one for the prior policy.

With the posterior policy in place, we can repeat steps \ref{ln:select_initial_condition} through \ref{ln:do_posterior_sos_analysis} to update the posterior model any number of times before stopping. Supposing that we perform $T$ iterations of this process, the final output of the algorithm will be the posterior model, the posterior policy $\kappa^T$, and the posterior ROA estimate $\mathcal{R}^T$.

\begin{algorithm}%
    \caption{Bayesian Safe ROA Learning with a polynomial GP model}
    \label{alg:sos_gp}
    \KwIn{Base model $f(x)+g(x)u$; prior kernel degree $p$ and hyperparameters $\{\alpha_i^2\}_{i=1}^p$; GP regression noise parameter $\sigma_n^2$; number $T$ of iterations.}
    \KwOut{Posterior control policy $\kappa^T$; posterior Lyapunov function $V^T$; posterior ROA $\mathcal{R}^T = \{x\in\R^{n_x} | V^T(x) \le \gamma^T\}$}

Construct the prior model $\dot{x}=f(x)+g(x)u+\hat{w}(x)$, where $\hat{w}(x)$ is a GP with mean zero and kernel $k(x,y)= \alpha_1^2(x^\top y)+\dotso+\alpha_p^2(x^\top y)^p$. Construct an empty data set $\mathcal{D}^0=\{\}$ \label{ln:construct_prior_model}\;
Solve the SOS program described in \eqref{eq:sos_ROA} using the prior model to compute the prior policy $\kappa^0$, prior Lyapunov function $V^0$, and prior ROA $\mathcal{R}^0$ \label{ln:do_prior_sos_analysis}\;
\For{$i \in \{1,\dotsc,T\}$} {
Select an initial condition $x_0^{i-1}\in\mathcal{R}^{i-1}$ \label{ln:select_initial_condition}\;
Collect data $\{x^{(j)}\}_{j=1}^{N^i}$ and $\{\dot{x}^{(j)}\}_{j=1}^{N^i}$ of $N^i$ points from a trajectory on the true dynamics with initial condition $x_0^{i-1}$, using the exploration policy $\kappa_e^i$ defined in (\ref{eq:exploration_policy}). Add this to the data set, setting $\mathcal{D}^i = \mathcal{D}^{i-1}\cup \{(x^{(j)}, \dot{x}^{(j)})\}_{j=1}^{N^i}$ \label{ln:get_data}\;
Use (\ref{eq:gp_mean}) and (\ref{eq:gp_var}) to compute the mean and variance of the GP with the data set $\mathcal{D}^i$ \label{ln:do_gp_regression}\;
Solve the SOS program described in \eqref{eq:sos_ROA} using the posterior model to compute the posterior policy $\kappa^i$, posterior Lyapunov function $V^i$, and posterior ROA $\mathcal{R}^i$ \label{ln:do_posterior_sos_analysis}\;
}
\end{algorithm}	

\section{Example: Inverted Pendulum with Input Saturation}%
\label{subsec:example_inverted_pendulum}

In this section, we demonstrate Algorithm \ref{alg:sos_gp} by using it to investigate the dynamics near the unstable equilibrium of a two-state inverted pendulum model.
The pendulum model, adapted from~\cite{berkenkamp2016safe}, includes an input saturation which prevents the system from being globally stabilized. Reference \cite{berkenkamp2016safe} analyzes the stability of this system for a fixed policy and Lyapunov function determined from a linearized model, and uses a GP with a non-polynomial kernel to verify a sublevel set of the fixed Lyapunov function as a ROA estimate. For our analysis, we do not need a prior safe policy and Lyapunov function to be given: we instead take a prior model
(also based on a linearization) and a kernel function, and use it to synthesize a prior controller and a ROA inner-approximation. We then collect a trajectory inside the prior ROA using the exploration policy, and use this data to compute a posterior model and synthesize a new policy and ROA. 

The true system dynamics for the pendulum are
\begin{equation}
\label{eq:pendulum_true_dynamics}
\begin{split}
    \dot{x}_1 &= x_2\\
    \dot{x}_2 &= \tfrac{g}{\ell}\sin(x_1)
               - \tfrac{\mu}{M\ell^2}x_2
               + \tfrac{1}{M\ell^2}\text{Sat}(u),
\end{split}
\end{equation}
where $M$ is the mass of the pendulum, $\ell$ is its length, and $g$ is gravitational acceleration.
The coordinates are chosen so that $x_1=0, x_2=0$ is the unstable equilibrium. 
The $\text{Sat}(\cdot)$ function limits the input action to stay within  the range $[-Mg\ell\sin(30^\circ), Mg\ell\sin(30^\circ)]$. With this input saturation in place, the inverted pendulum cannot return to the upright position once it deviates from upright by more than 30 degrees.

The input saturation also means that this system is not input-affine. To remedy this, we use the formulation from the Remark in Section \ref{sec:preliminaries}: we introduce an auxiliary \emph{input state} $x_u$ and augment (\ref{eq:pendulum_true_dynamics}) to
\begin{equation}
\label{eq:pendulum_augmented_dynamics}
\begin{split}
    \dot{x}_1 &= x_2\\
    \dot{x}_2 &= \tfrac{g}{\ell}\sin(x_1)
               - \tfrac{\mu}{M\ell^2}x_2
               + \tfrac{1}{M\ell^2}\text{Sat}(x_u)\\
    \dot{x}_u &= v.
\end{split}
\end{equation}

We start the algorithm with the linearization of  (\ref{eq:pendulum_augmented_dynamics}); that is, we take
\begin{equation}
\label{eq:pendulum_linearized_dynamics}
f(x) = 
\begin{bmatrix}
x_2 \\ \tfrac{g}{\ell}x_1 - \tfrac{\mu}{M\ell^2}x_2 + \tfrac{1}{M\ell^2}x_u\\ 0
\end{bmatrix},
\quad
g(x) = 
\begin{bmatrix}
0 \\ 0 \\ 1
\end{bmatrix}
\end{equation}
as the inputs $f$ and $g$ to Algorithm \ref{alg:sos_gp}.
For the prior kernel, we use the degree 3 kernel
\begin{equation}
    \label{eq:pendulum_prior_kernel}
    k(x,y) = \alpha_1^2(x^\top y) + \alpha_2^2(x^\top y)^2 + \alpha_3^2(x^\top y)^3.
\end{equation}
Since the dynamics of $x_1$ are purely kinematic, we can assume that the given model is accurate. Similarly, since $x_u$ is a constructed state, we can assume its dynamics are accurate. Therefore, we assume that the vector of unknown dynamics has the form $w(x) = [0~~w_2(x)~~0]^\top $, requiring only one GP model for the unknown dynamics of $x_2$.

To complete the prior model, we select kernel hyperparameters for $w_2(x)$.
We will take as prior knowledge that our linearization is accurate, and that the nonlinear terms contain a strong odd component. We incorporate this knowledge into the prior model by setting $\alpha_1^2$ and $\alpha_2^2$ to a small value, 
namely $\alpha_1^2=\alpha_2^2=0.075$. Since we have no further prior knowledge of the third-order term, we will set $\alpha_3^2$ to be larger than $\alpha_1$ and $\alpha_2$, namely $\alpha_3^2=1.5$. We will also assume that our $\dot{x}$ measurements, taken from a finite-difference approximation on the observed states, are reasonably accurate, and use this knowledge by setting the GP regression noise parameter $\sigma_n^2$ to a low value value, namely $\sigma_n^2=0.01$.

\begin{figure}[h]
    \centering
    \includegraphics[width=0.5\textwidth]{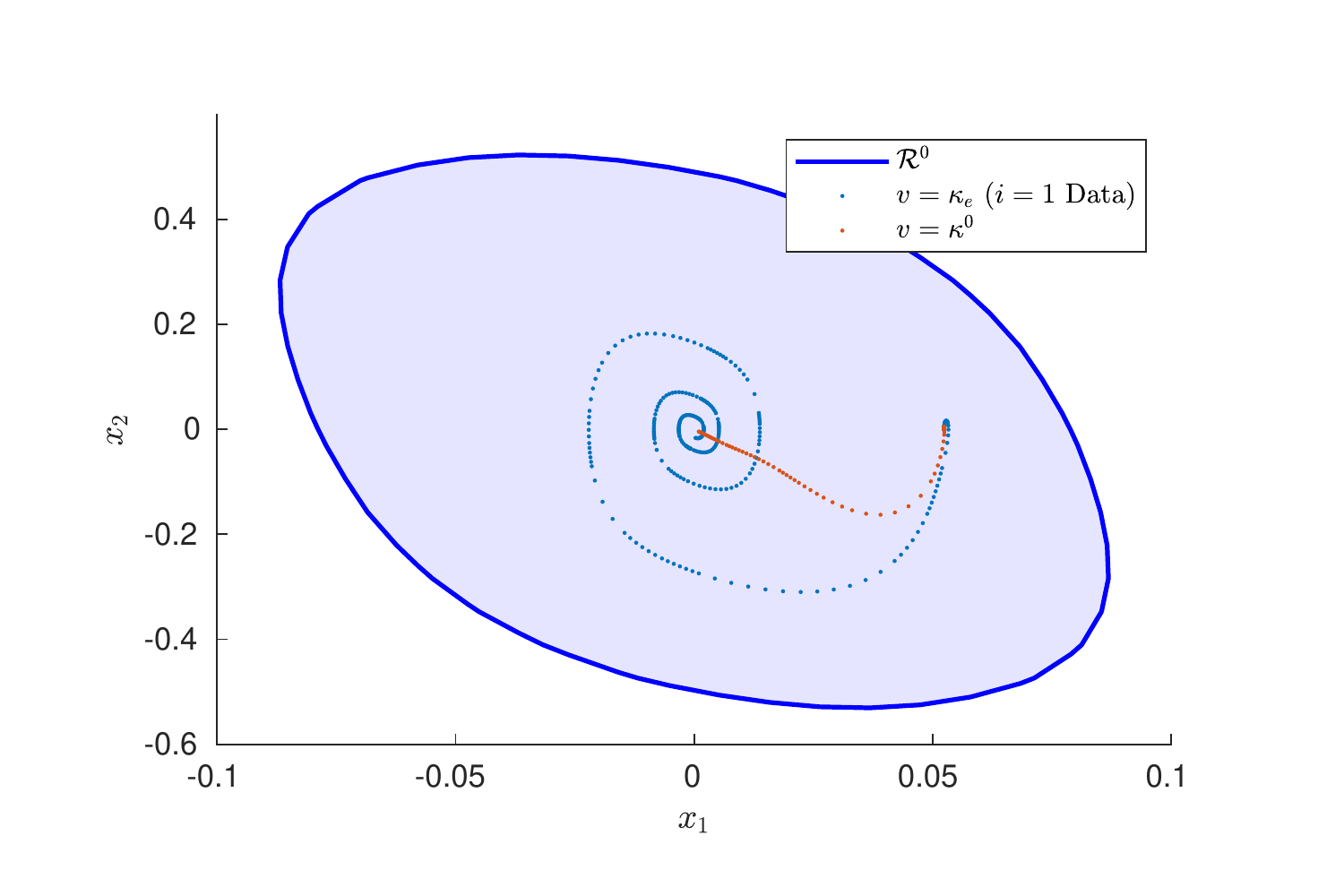}
    \caption{The prior ROA computed using the prior system (\ref{eq:pendulum_linearized_dynamics}) and prior kernel (\ref{eq:pendulum_prior_kernel}), projected onto $x_1$ and $x_2$. Two trajectories are also shown using the two prior control policies, the base SOS policy $v(t)=\kappa(x(t))$ and the exploration policy $v(t)=\kappa_e(x(t))$. The exploration policy visits more of the state space than the base policy.}
    \label{fig:pendulum_prior_roa}
\end{figure}

Figure \ref{fig:pendulum_prior_roa} shows the results of lines \ref{ln:do_prior_sos_analysis}, \ref{ln:select_initial_condition}, and \ref{ln:get_data} of Algorithm \ref{alg:sos_gp} using the selected $f$, $g$, and $k$. The decision variables in the SOS analysis---$V^0$, $\kappa^0$, and the S-procedure certificates---$s_V, s_{d,i}, s_\gamma$ are degree 4 polynomials, and we take $\eta=3$.
The prior ROA certifies that the prior policy $\kappa^0$ can restore angle deviations in the range of about $\pm 4.5$ degrees, starting from rest, in the presence of any $w_2$ that is bounded above and below by $-\eta\sigma_2(x) \le w_2(x) \le \eta\sigma_2(x)$, where $\sigma_2(x) = \sqrt{k(x,x)}$ is the prior variance.

For step \ref{ln:select_initial_condition}, we select an initial condition which starts from rest with an initial angle deviation of $3^\circ$; that is, we take $x_1=3\pi/180$, $x_2=x_u=0$. Figure \ref{fig:pendulum_prior_roa} shows data from two trajectories on the true dynamics with this initial condition. One trajectory, following step \ref{ln:get_data} of Algorithm \ref{alg:sos_gp}, uses the exploration policy $v(t)=\kappa_e(x(t))$. This trajectory will be used as the data set $\mathcal{D}^0$ for the next step. 
The other trajectory uses the prior policy $\kappa^0$, from the same initial condition chosen in step \ref{ln:select_initial_condition}. The exploration policy
provides data from a wider area of the prior ROA before settling to the equilibrium, by allowing more transients to remain than the prior policy.
\begin{figure}[h]
    \centering
    \includegraphics[width=0.45\textwidth]{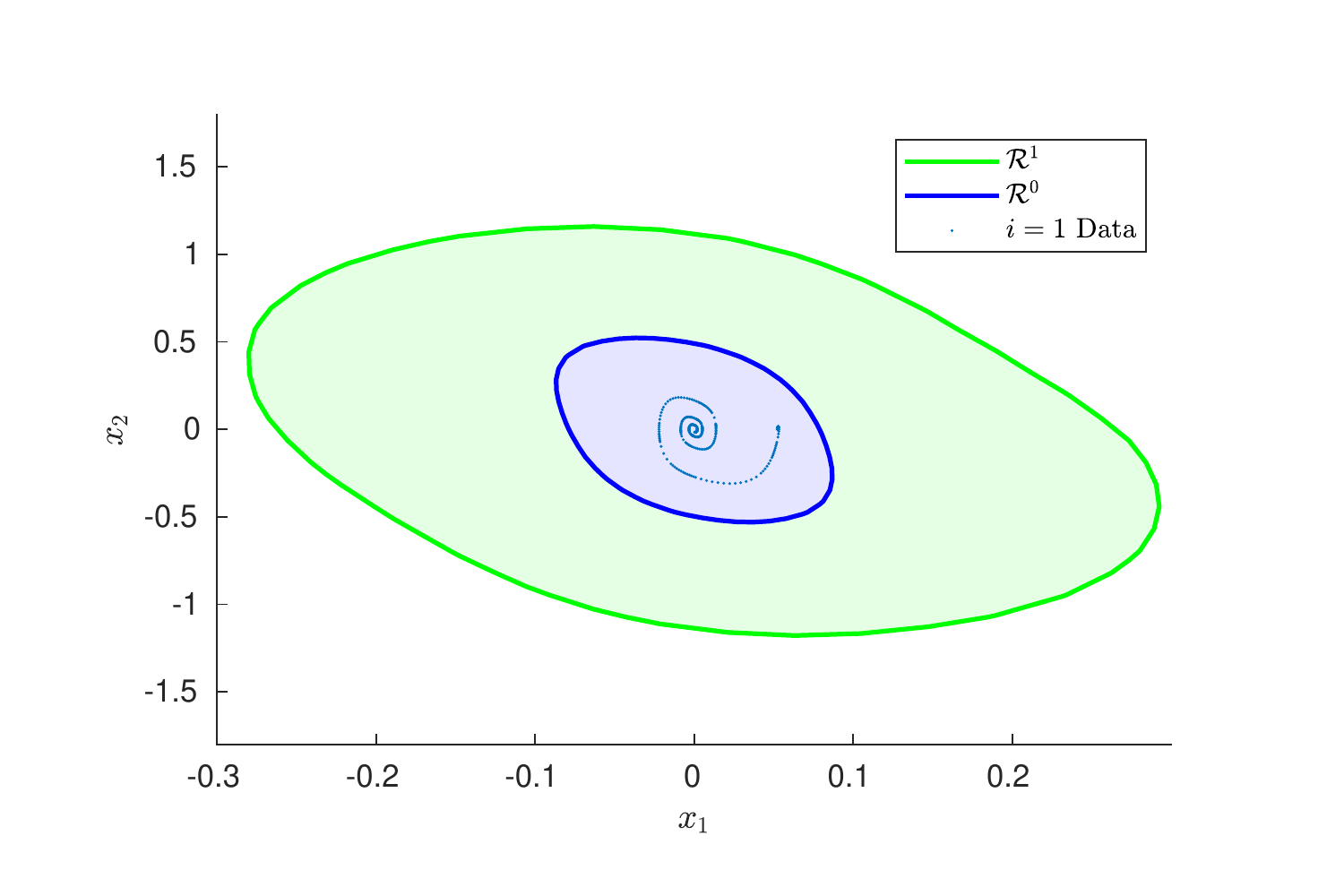}
    \caption{The posterior ROA, projected onto $x_1$ and $x_2$. The posterior model incorporates the data collected by the exploration trajectory from iteration $i=1$ of Algorithm \ref{alg:sos_gp}.}
    \label{fig:pendulum_posterior_roa}
\end{figure}

Figure \ref{fig:pendulum_posterior_roa} shows the posterior ROA $\mathcal{R}^1$ computed by step \ref{ln:do_posterior_sos_analysis} of Algorithm \ref{alg:sos_gp}. The posterior model was computed using the data set $\mathcal{D}^0$ comprising the $x$ data points from the prior trajectory with the exploration policy and finite-difference approximations for $\dot{x}$. With $T=1$, this is the final step of the algorithm. By incorporating the trajectory data, the posterior analysis successfully extends the size of the ROA. In particular, the range of safe angle deviations from rest is extended to $\pm 16.5$ degrees. 

There are two mechanisms which allow the posterior ROA to be larger. First, the posterior model more closely matches the true dynamics than the prior, since it includes higher-order terms that are fit from data. Second, the posterior variance is less than the prior variance at all points in the state space. Since the variance determines the constraints on $w$ in the SOS problem, the posterior controller can be robust against a smaller class of unknowns than the prior model while upholding the same probabilistic guarantee.
\begin{figure}[h]
    \centering
    \includegraphics[width=0.5\textwidth]{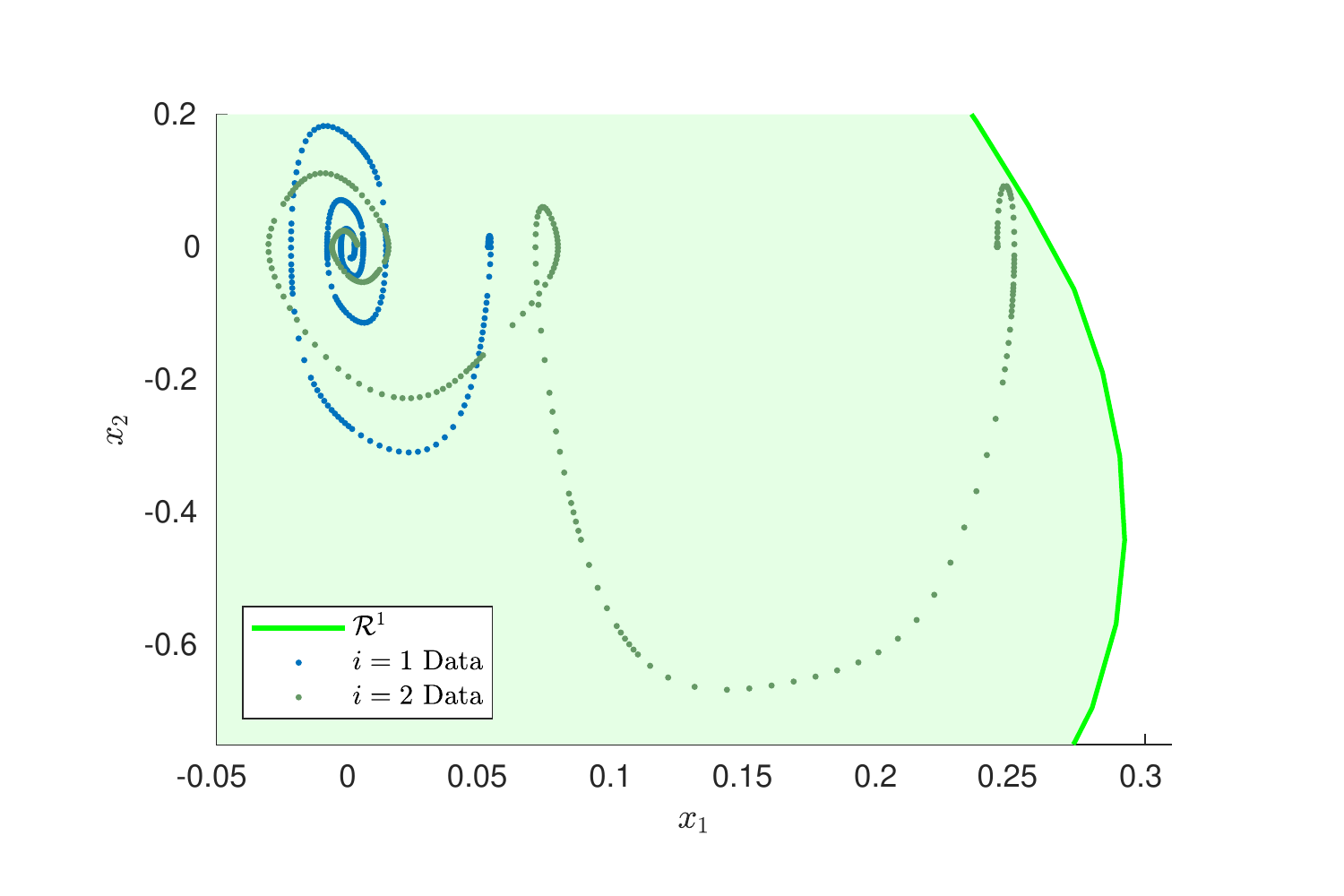}
    \caption{The exploration trajectory from iteration $i=2$ of Algorithm \ref{alg:sos_gp}, projected onto $x_1$ and $x_2$. The objective of (\ref{eq:exploration_policy}) encourages the exploration policy to avoid the data from previous trajectories.}
    \label{fig:pendulum_posterior-trajectory}
\end{figure}

Though only one trajectory is needed to complete one iteration of Algorithm \ref{alg:sos_gp}, we demonstrate how the exploration policy responds to data from previous iterations by simulating an exploration trajectory for the $i=2$ iteration. 
To that end, we pick an initial condition in the posterior ROA, starting from rest with an initial angle deviation of $14^\circ$, and simulate a trajectory on the true dynamics using the exploration policy guided by the posterior variance. 

The resulting trajectory is shown in Figure \ref{fig:pendulum_posterior-trajectory}.
Recall that, by maximizing the objective in (\ref{eq:exploration_policy}), the exploration policy is encouraged to increase the total variance of states that the system visits. The exploration policy increases the information gained by the $i=2$ trajectory in two ways. First, it avoids the $i=1$ trajectory as much as possible, so that it does not collect data in parts of the state space that have already been visited. Second, it adds several excitations into the trajectory---once at the start of the trajectory, and again near $(x_1=.075,x_2=0)$-- where it briefly reverses direction, increasing the amount of time the system can explore before settling into the equilibrium.

\section{Conclusion}%
\label{sec:conclusion}

The proposed method can take an initial prior model for the dynamics of a system and improve the model using data, while ensuring that the process of collecting takes place in a safe ROA. Since the GP model learns an estimate for the system dynamics in closed-form with quantified uncertainty, the learned model can be guaranteed safe with high confidence.
The proposed method lifts two limitations faced by earlier work in safe learning with GPs. First, we are not restricted to a fixed, given policy and Lyapunov function: using polynomial kernel functions allows for policies and Lyapunov functions to be synthesized by SOS analysis. Second, we do not need to assume the existence of an \emph{a priori} safe controller to initialize the safe learning process: by establishing prior information into a Bayesian prior model, we can compute an exploration controller which is guaranteed to be safe on the prior model dynamics.

However, the restriction to polynomial kernels places a limit on the types of unknown dynamics the system can learn. In particular, the condition that the unknown dynamics be well-approximated by a polynomial prohibits the method from learning dynamics with discrete transitions or discontinuities. Extending the method to work on a larger class of dynamics would increase the utility of the method. Another useful extension would be to allow for the synthesis of other types of safety guarantees than ROAs for learned systems, for instance barrier certificates or reachable sets.

\bibliographystyle{IEEEtran}
\bibliography{IEEEabrv,refs}
\appendix
\section{Appendices}
\subsection{Proof for Proposition~\ref{prop:krause_inequality_holds}}\label{app:prfprop4}
\begin{proof}
By Assumption \ref{a:uncertainty_approximated_by_polynomial}, we know that in a region $\mathcal{X}$ containing the origin, each $w_i$ can be approximated by a polynomial $q_i$ with uniform error $\epsilon$. The measurements of $w_i$ are effectively measurements of $q_i$ corrupted by this uniformly-bounded noise.
Let $p_{q_i}$ be the degree of this polynomial, and let $k_i(x,y)$ be a kernel of the form (\ref{eq:poly_kernel_zao}) with $p\ge p_{q_i}$. Then $q_i\in\mathcal{H}(k_i)$, by Propositions \ref{prop:rkhs_homogeneous} and \ref{prop:rkhs_sum}. Since $\mathcal{H}(k_i)$ is finite-dimensional, the norm $\norm{q_i}_{k_i}$ is finite.
By assumption \ref{a:bounded_noise}, the measurements of $w_i$ are also subject to an additional noise uniformly bounded by $\sigma_n$. Since the measurements of $w_i$ act as measurements of $q_i$ with noise uniformly bounded by $\sigma_n + \epsilon$, and $\norm{q_i}_{k_i}\le\infty$, the function $q_i$ and the kernel $k_i$ satisfy the assumptions of Lemma \ref{lem:krause_inequality}.

\end{proof}

\subsection{Iterative algorithm for solving the SOS problem \eqref{eq:sos_ROA}}\label{app:iteralgo}
\begin{algorithm} [h]
	\caption{Iterative method for solving \eqref{eq:sos_ROA}}
	\label{alg:solve_sos}
		\KwIn{function $\bar V$ such that constraints \eqref{eq:sos_ROA} are feasible by proper choice of $s_V, s_{d,i}, \kappa, \gamma$.}
		\KwOut{($\kappa$, $\gamma$, $V$) such that with the volume of $\mathcal{R}$ having been enlarged.}
		\For{$j \in \{1,...,N_{iter}$\}}{
		$\boldsymbol{\gamma}$\textbf{-step}: decision variables $(s_V, s_{d,i}, \kappa,\gamma)$.
		Maximize $\gamma$ subject to \eqref{eq:sos_ROA}  using $V = \bar V$. 
		This yields ($\bar s_V, \bar \kappa$) and optimal reward $\bar \gamma$.
		
		$\boldsymbol{V}\textbf{-step}$: decision variables $(s_\gamma, s_{d,i},  V)$; 
		Maximize the feasibility subject to \eqref{eq:sos_ROA} as well as 
			$s_\gamma \in \Sigma[x]$, and
			\begin{align}
			(\bar \gamma - V) +  (\bar V - \bar \gamma) s_\gamma \in \Sigma[x], \label{eq:levelset_grow}
			\end{align}
			using $\gamma = \bar \gamma, s_V = \bar s_V,  \kappa=\bar \kappa$. This 
			yields $\bar V$.
		}
\end{algorithm}
A linear state feedback for the linearization of $f$ and $g$ about the origin is used to compute the initial iterate, $\bar V$.
\end{document}